\documentclass[reqno]{amsart}

\usepackage[top=1in, bottom=1in, left=1in, right=1in]{geometry}
\usepackage{tipa}
\usepackage{amssymb}
\usepackage{graphicx}
\usepackage{hyperref}
\usepackage{enumerate}
\usepackage{xcolor}
\usepackage{color}
\usepackage{mathrsfs}     %可以打出花体的 ABCD，例如： 域流 F 
\usepackage{dsfont}
\usepackage{makecell}    %做表格用
\usepackage[format=plain,justification=centering,singlelinecheck=false,]{caption}   %设置caption
\usepackage[FIGTOPCAP]{subfigure} % Figure 子标题放在上面

\usepackage{amsmath}    %多行公式可以换页
\allowdisplaybreaks[4]

\def\be{\begin{equation}}
\def\ee{\end{equation}}

\newtheorem{theorem}{Theorem}[section]
\newtheorem{lemma}[theorem]{Lemma}
\newtheorem{prop}[theorem]{Proposition}

\newtheorem{remark}{Remark}[section]

\numberwithin{equation}{section}

\newcommand{\abs}[1]{\left|#1\right|}                     %绝对值设置%
\begin{document}

\title[Consistent Pricing of VIX and Equity Derivatives with the 4/2 Model]{Consistent Pricing of VIX and Equity Derivatives with the 4/2 Stochastic Volatility Plus Jumps Model}

\author[W. Lin]{Wei Lin}
\address{School of Mathematical Sciences , Zhejiang University, Hangzhou, 310027, People's Republic of China}
\email{weilin1991@zju.edu.cn; mathslin@126.com}
\thanks{}

\author[S. H. Li]{Shenghong Li}
\address{School of Mathematical Sciences , Zhejiang University, Hangzhou, 310027, People's Republic of China}
\email{shli@zju.edu.cn}

\author[X. G. Luo]{Xingguo Luo}
\address{College of Economics, Zhejiang University, Hangzhou, 310027, People's Republic of China}
\email{xgluo@zju.edu.cn}

\author[S. Chern]{Shane Chern}
\address{School of Mathematical Sciences , Zhejiang University, Hangzhou, 310027, People's Republic of China}
\email{shanechern@zju.edu.cn; chenxiaohang92@gmail.com}

%    General info

\date{}

\dedicatory{}


\begin{abstract}
%This is the final exam review materials for modern analysis.
%
In this paper, we develop a 4/2 stochastic volatility plus jumps model, namely, a new stochastic volatility model including the Heston model and 3/2 model as special cases. Our model is highly tractable by applying the Lie symmetries theory for PDEs, which means that the pricing procedure can be performed efficiently. In fact, we obtain a closed-form solution for the joint Fourier-Laplace transform so that equity and realized-variance derivatives can be priced. We also employ our model to consistently price equity and VIX derivatives. In this process, the quasi-closed-form solutions for future and option prices are derived.  Furthermore, through adopting data on daily VIX future and option prices, we investigate our model along with the Heston model and 3/2 model and compare their different performance in practice. Our result illustrates that the 4/2 model with an instantaneous volatility of the form $(a\sqrt{V_t}+b/\sqrt{V_t})$ for some constants $a, b$ presents considerable advantages in pricing VIX derivatives.
\end{abstract}

\maketitle

%              \section*{不带编号的标题}      %%%%%%               标题不带编号就这样实现
\section{Introduction}
Since the Chicago Board Options Exchange (CBOE) launched the CBOE Volatility Index (VIX) futures in March 2004 and later VIX options in February 2006, the trading volume of derivatives on the VIX index has grown considerably over the last decade and become popular among investors. One reason is that VIX derivatives provide investors with a mechanism to directly and effectively invest in the volatility of the S\&P500 index without having to factor in the price changes of the underlying instrument, dividends, interest rates or time to expiration. Moreover VIX derivatives are the first of an entire family of volatility products to be traded on exchanges. The index is also known as the ``fear gauge" as in terms of market turmoil and large price movements, for the VIX index tends to rise, whereas when the market is easing upward in a long-run bull market, the VIX index remains low and steady. Naturally, this development has fueled the demand for models that are capable of simultaneously reproducing the observed characteristics of products on both indices, since derivatives products are traded on both the underlying index and the volatility index. Here, model that are able to capture these joint characteristics are known as consistent models. 

A growing body of literature has emerged on the joint modeling of equity and VIX derivatives. One approach is adopted in Zhang and Zhu \cite{ZZ2006}, Sepp \cite{S2008}, Zhu and Lian \cite{ZL2012}, Lian and Zhu \cite{LZ2013} and Baldeaux and Badran \cite{BB2014}. The discounted price of derivatives can be expressed as a strict local martingale once the instantaneous dynamics of the underlying index are specified under the (putative) risk neutral probability measure $\mathbb{Q}$. Zhang and Zhu \cite{ZZ2006} derived an analytic formula for VIX futures under the assumption that the S\&P500 is modeled by Heston \cite{H1993}. Baldeaux and Badran \cite{BB2014} come up with more general formulae which allow for an empirical analysis to be performed to assess the appropriateness of the 3/2 framework the consistent pricing of equity and VIX derivatives. The Heston model \cite{H1993} takes the instantaneous variance as a mean reverting squared Bessel process (usually called CIR or square root process since it displays a power 1/2 in the diffusion term), while the subsequent 3/2 model of Heston \cite{H1997} and Platen \cite{P1997} which is the inverse of a CIR process and, is still mean reverting with a power 3/2 in the diffusion term. Meanwhile, Grasselli \cite{G2014} mentioned a less-known stochastic volatility model that combines as special cases the classic Heston with 3/2 model. This model considers as the superposition of the 1/2 and the 3/2 terms, i.e., we introduce an instantaneous volatility of the form $\left(a\sqrt{V_t}+b/\sqrt{V_t}\right)$ for some constants a, b, where $V_t$ is the CIR process. Although these authors provided exact solution by characteristic function approach for the price of VIX derivatives when the S\&P500 is modeled by either Heston model or 3/2 model with simultaneous jumps in the underlying index, it remains to be shown whether 4/2 model is able to price the equity and VIX derivatives consistently. This paper aims to fill this vacuum.

Heston model has been justified a successful model in literature and in the banking industry for many reasons, e.g., smile and skew to be reproduced with parsimonious number of parameters, clear financial meaning on each parameter and its tractability. In addition, under a certain parameter restriction (the Feller condition), the volatility process remains strictly positive, which constitutes a nice property of the model. However, some shortcomings have been shown immediately in calibration of the model on real data. Feller condition is often violated because a high volatility-of-volatility parameter is required to fit the steep skews in equity markets. Moreover, when instantaneous volatility increases, the skew will flatten. Then, the Heston model assigns significant weight to very low and vanishing volatility scenarios and is unable to produce extreme paths with high volatility of volatility. Having said that, the Heston model still remains a good benchmark for any-single-factor stochastic volatility model that can be quickly calibrated on the market. 

The selection of the inverse of a CIR model (3/2 model) for the underlying index is motivated by several observations in recent literatures. Compared with Heston model, both empirical and theoretical evidences suggest that the 3/2 model is a reasonable candidate for modeling instantaneous variance due to quick reversion when the process is high. Baldeaux and Badran \cite{BB2014} presented joint modeling of equity and VIX derivatives when the underlying index follows a 3/2 process with jumps in the index only. Following his conclusions, for 3/2 model, the implied volatility of VIX option are upward-sloping, which were consistent with market data. In fact, in the Heston model, the short-term skew flattens when the instantaneous variance increases, whereas in the 3/2 model, the short-term skew steepens when the instantaneous variance increases. Finally, applying to Fourier methodology, 3/2 model also remains to obtain an acceptable level of tractability when pricing derivatives products since closed form characteristic functions consist special functions like the hypergeometric confluent and Gamma functions.  

There are two main contributions of this paper. Under the assumption that underlying index follows the 4/2 plus jumps model, one contribution is the derivation of closed-form solution for Fourier-Laplace transform of log-equity and realized-variance. The other is the derivation of quasi-closed-form solutions for the pricing of VIX future and option. Our model combines the properties of both the Heston and 3/2 models. The 4/2 model goes beyond classic models: it behaves as a two-factor model where the stochastic factors $\sqrt{V_t}$ and $1/\sqrt{V_t}$ are closely related (they are indeed perfectly correlated) but still maintaining different properties in explaining the implied volatility surface, as preciously illustrated. In order to capture features of implied volatility in equity option for short maturities, jumps are needed in the underlying index. We also discuss whether 4/2 model plus jumps is a martingale and derive the conditions so that the discounted stock price is a martingale under the pricing measure. Despite the simplicity of the financial framework, the 4/2 model plus jumps leads a problem of computing an expectation that goes beyond the known results on squared Bessel processes. So we provide an explicit solution to the problem by using  results of Grasselli \cite{G2014} based on the theory of Lie symmetries for PDEs. Finally, we use this approach to obtain a quasi-closed-form solution for VIX futures and options in the 4/2 model plus jumps, slightly extending of the stochastic volatility pricing formula presented in Zhang and Zhu \cite{ZZ2006}, Lian and Zhu \cite{LZ2013}, Baldeaux and Badran \cite{BB2014} and Grasselli \cite{G2014}.

The structure of the paper is as follows: The 4/2 model plus jumps is introduced in Section 2. Next in Section 3, we investigate the martingale property of the discount asset price and establish the conditions which ensure that the discounted stock price is a martingale under the assumed pricing measure. Characteristic functions for the logarithm of the index and the realized variance are derived in Section 4. Furthermore, the quasi-closed formulae for futures and options price on the VIX is found in Section 5. The data to be used in the empirical analysis is described in Section 6. Then in Section 7, model parameters are then estimated using historical VIX data. VIX options pricing formulae are tested against the market data in Section 8. Finally, Section 9 provides a brief conclusion.

\section{The 4/2 Plus Jumps Model}\label{sec1}
In this section, we introduce the 4/2 stochastic volatility plus jumps model. On a probability space ($\Omega$, $\mathscr{F}$, $\mathbb{Q}$), consider the risk-neutral dynamics for the stock price with non-dividend-paying and the variance processes according to the following SDE:
\begin{equation}\label{eq2.1}
\frac{dS_t}{S_t}=\left(r-\lambda \tilde{\mu}\right)dt+\left(a\sqrt{V_t}+\frac{b}{\sqrt{V_t}}\right)dZ_t+\left(e^{\xi}-1\right)dN_t,
\end{equation}
where the stochastic factor $V$ evolves as
\begin{equation}\label{eq2.2}
dV_t=\kappa \left(\theta-V_t\right) dt+\sigma \sqrt{V_t}dW_t,
\end{equation}
with $r,\kappa, \theta, \in \mathbb{R}_+; a, b, \sigma \in \mathbb{R}$. The Brownian motion $Z,W$ are defined on a filtered probability space $(\Omega, \mathscr{F}, (\mathscr{F}_{t\ge0}),\mathbb{Q})$ and correlated through $d\langle W,Z \rangle_t=\rho dt$ and $V_0=v \in \mathbb{R}_+$. As usual, $\rho$ satisfies $-1\le \rho \le 1$. We denote  by $N$ a Possion process at constant rate $\lambda$, by $e^{\xi}$ the relative jump size of the stock and $N$ is adapted to a filtration $(\mathscr{F}_t)_{\in[0,T]}$. The distribution of $\xi$ is assumed to be normal with mean $\mu$ and variance $\eta^2$, where the parameters $\mu,\ \tilde{\mu}$ and $\eta$ satisfy the following relationship:

\begin{equation}\label{eq2.3}
\tilde{\mu}=\exp\left({\mu+\frac{1}{2}\eta^2}\right)-1
\end{equation}

Here $r$ stands for the riskless interest rate, i.e we identify $\mathbb{Q}$ with the risk neutral probability measure. From \eqref{eq2.2} we can recover the Heston stochastic volatility model and the 3/2 model by taking $b=0$ and $a=0$ respectively. In fact, note that taking $X_t={V_t}^{-1}$, it follows

\begin{equation}\label{eq2.4}
dX_t= \tilde{\kappa} X_t \left( \tilde{\theta}-X_t \right) dt+\tilde{\sigma} {X_t}^{\frac{3}{2}}dW_t,
\end{equation}
where
\begin{align*}
\tilde{\kappa}=\kappa \theta - \sigma^2,\quad 
\tilde{\theta}=\frac{\kappa}{\kappa \theta-\sigma^2},\quad
\tilde{\sigma}=-\sigma.
\end{align*}

It is worth noting that unlike the Heston model, the above model has a non-linear drift. The speed of mean reversion is not constant, as is the case for the Heston model, but is now a stochastic quantity and is proportional to the instantaneous variance.

Integrating Eq.\eqref{eq2.1} and splitting the Brownian motion $Z$ into $W$ and its orthogonal part $W^{\perp}$ yields
\begin{equation}
S_t=\tilde{S_t} \prod_{j=1}^{N_t}e^{\xi_j},
\end{equation}
where

\begin{align*}
\tilde{S_t}=&S_0\exp \Bigg(\left(r-\lambda \tilde{\mu}\right)t-\frac{1}{2}\int_{0}^{t}\left(a\sqrt{V_s}+\frac{b}{\sqrt{V_s}}\right)^2dt+\rho \int_{0}^{t}\left(a\sqrt{V_s}+\frac{b}{\sqrt{V_s}}\right)dW_s\\
&+\sqrt{1-\rho^2}\int_{0}^{t}\left(a\sqrt{V_s}+\frac{b}{\sqrt{V_s}}\right)dW_s^{\perp}\Bigg)
\end{align*}
and $\xi_j$ denote the logarithm of the relatives jump size of the $j$th jump.

\section{Strict Local Martingale Property of the Discounted Asset}
Even though the model in \eqref{eq2.1} and \eqref{eq2.2} is not affine, we can now determine if the discounted stock price is a martingale under our assumed pricing measure. In order to see whether the process $\bar{S_t}=\frac{S_t}{e^{rt}}$ is a martingale, and not just a local martingale, the Feller non-explosion test for $V_t$ must be satisfied under both historical and risk neutral probability measures. This condition has been first established by Drimus \cite{D2012} who imposed the martingale property to the discounted asset in his calibration.

\begin{prop}
Let $S$ and $V$ be given by Eqs. \eqref{eq2.1} and \eqref{eq2.2} respectively. Then the discounted stock price $\bar{S_t}=\frac{S_t}{e^{rt}}$ is a martingale, and not just a local martingale under $\mathbb{Q}$ ,if and only if 
\begin{equation}
2\kappa\theta+2\rho\sigma b < \sigma^2 \le 2\kappa\theta.
\end{equation}

\begin{proof}
\begin{align}\label{eq3.2}
&\mathbb{E}\left(\bar{S}_T | F_t\right)\notag\\
=&\bar{S_t}\mathbb{E}\left[\exp\left(-\frac{1}{2}\int_{t}^{T}\left(a\sqrt{V_s}+\frac{b}{\sqrt{V_s}}\right)^2ds+\rho \int_{t}^{T} \left(a\sqrt{V_s}+\frac{b}{\sqrt{V_s}}\right)dW_s\right.\right.\notag\\
&\quad\left.\left.+\sqrt{1-\rho^2}\int_{t}^{T}\left(a\sqrt{V_s}+\frac{b}{\sqrt{V_s}}\right)dW_s^{\perp}\right) \Bigg|F_t\right]\times \mathbb{E}\left(\prod_{j=N_t+1}^{N_T}e^{\xi_j}\right) e^{-\lambda \tilde{\mu}(T-t)}\notag\\
=&\bar{S_t}\mathbb{E}\left[ \exp \left(-\frac{1}{2}\int_{t}^{T}\left(a\sqrt{V_s}+\frac{b}{\sqrt{V_s}}\right)^2ds+\int_{t}^{T}   \left(a\sqrt{V_s}+\frac{b}{\sqrt{V_s}}\right)dZ_s\right) \right]\notag\\
=&\bar{S_t}\mathbb{E}\left[ \exp \left(-\frac{\rho^2}{2}\int_{t}^{T}\left(a\sqrt{V_s}+\frac{b}{\sqrt{V_s}}\right)^2ds+\rho\int_{t}^{T}\left(a\sqrt{V_s}+\frac{b}{\sqrt{V_s}}\right)dW_s\right) \right]\\
=&\bar{S_t}\mathbb{E}\left[\zeta_{t,T}\right]\notag
\end{align}
where we define the exponential local martingale process $\zeta_{t,T}=\{\zeta_s,t\le s\le T\}$ via
\begin{equation}
\zeta_t:=\exp \left(-\frac{\rho^2}{2}\int_{0}^{t}\left(a\sqrt{V_s}+\frac{b}{\sqrt{V_s}}\right)^2ds+\rho\int_{0}^{t}\left(a\sqrt{V_s}+\frac{b}{\sqrt{V_s}}\right)dW_s\right)
\end{equation}

Equation \eqref{eq3.2} is clearly independent of the jump component of $S$. Hence $\bar{S}$ is martingale under $\mathbb{Q}$ when the Feller non-explosion test for $V_t$ must be satisfied under both historical and risk neutral probability measures. Since this question was answered in Grasselli \cite{G2014}, see his equation (8), (9) and (10), the desired result follows.
\end{proof}
\end{prop}

\section{Equity and Realized-Variance Derivatives}
In this section, we derive formulae for the pricing of equity and realized-variance derivatives under the 4/2 plus jumps model. We demonstrate that 3/2 plus jumps model is included in 4/2 plus jumps model so that a better fit to short-term smile can be obtained (Baldeaux and Badran \cite{BB2014}). Moreover, 4/2 plus jumps model is still analytic tractable. 

Consider
\begin{equation}
X_t:=\log S_t,\ t\ge0;
\end{equation}
and define the realized variance as the quadratic variation of $X$, viz.,
\begin{equation}
RV_T:=\int_0^T\left({a\sqrt{V_t}+\frac{b}{\sqrt{V_t}}}\right)^2dt+\sum_{j=1}^{N_T}(\xi_j)^2,
\end{equation}
where $RV_T$ denotes realized variance and $T$ denotes the maturity of time. We have the following Theorem \ref{thm4.1}, which is an extension of Proposition 3.2 in Baldeaux and Badran \cite{BB2014} and Proposition 1 in Grasselli \cite{G2014}.

\begin{theorem}\label{thm4.1}
Let $u\in \mathcal{D}_{t,T} \subset \mathbb{C}$ and $l \in \mathbb{R}^+$.In the 4/2 plus jumps model, the joint Fourier-Laplace transform of $X_T$ and $(RV_T-RV_t)$ is given by
\begin{align}
\mathbb{E}&\left[ \exp(uX_T-l(RV_T-RV_t)) \bigg| V_t,X_t\right]\notag\\
&\quad=e^{uX_t}\exp\left(\frac{\kappa^2\theta}{\sigma^2}\tau+u\left(r-\lambda\tilde{\mu}-ab-\frac{\kappa\theta a\rho}{\sigma}+\frac{\rho b\kappa}{\sigma}\right)\tau-2lab\tau+ab\tau(1-\rho^2)u^2\right)\notag\\
&\quad\quad\times\left(\frac{\sqrt{A_u}}{\sigma^2\sinh\left(\frac{\sqrt{A}\tau}{2}\right)}\right)^{m_u+1}V_t^{-\frac{bu\rho}{\sigma}-\frac{\kappa\theta}{\sigma^2}+\frac{1}{2}+\frac{m_u}{2}}\frac{\Gamma\left(\frac{bu\rho}{\sigma}+\frac{\kappa\theta}{\sigma^2}+\frac{1}{2}+\frac{m_u}{2}\right)}{\Gamma\left(m_u+1\right)}\notag\\
&\quad\quad\times\exp\left(\frac{V_t}{\sigma^2}\left(-au\rho\sigma-\sqrt{A_u}\coth\left(\frac{\sqrt{A_u}\tau}{2}\right)+\kappa\right)\right)\exp({\lambda\tau({c-1})})\notag\\
&\quad\quad\times_1F_1\left(\frac{bu\rho}{\sigma}+\frac{\kappa\theta}{\sigma^2}+\frac{1}{2}+\frac{m_u}{2},m_u+1,\frac{A_uV_t}{\sigma^4\sinh^2\left({\frac{\sqrt{A}\tau}{2}}\right)\left({K_u(t)-\frac{ua\rho}{\sigma}}\right)}\right)\notag\\
&\quad\quad\times\left({-\frac{ua\rho}{\sigma}+K_u({t})}\right)^{-\left({\frac{bu\rho}{\sigma}+\frac{\kappa\theta}{\sigma^2}+\frac{1}{2}+\frac{m_u}{2}}\right)}
\end{align}
where
\begin{align}
\tau&=T-t,\\
c&=\frac{\exp\left({\frac{2u\mu-2l\mu^2+u^2\eta^2}{2+4l\eta^2}}\right)}{\sqrt{1+2l\eta^2}},\\
A_u&=\kappa^2-2\sigma^2\left({\frac{\kappa au\rho}{\sigma}+\frac{1}{2}u^2({1-\rho^2})a^2-a^2\left({\frac{u}{2}+l}\right)}\right),\\
m_u&=\frac{2}{\sigma^2}\sqrt{\left({\kappa\theta-\frac{\sigma^2}{2}}\right)^2-2\sigma^2\left[\frac{bu\rho({\sigma^2-2\kappa\theta})}{2\sigma}+\frac{1}{2}u^2({1-\rho^2})b^2-b^2\left({\frac{u}{2}+l}\right) \right]},\\
K_u({t})&=\frac{1}{\sigma^2}\left({\sqrt{A_u}\coth\left({\frac{\sqrt{A_u}\tau}{2}}\right)+\kappa}\right).
\end{align}
Here $\Gamma$ and$\ _1F_1$ denote the Gamma function and hypergeometic confluent function, respectively. The transform is well defined for all $\tau \ge 0$ when the complex number u belongs to the strip $\mathcal{D}_{t,+\infty}=\mathcal{A}_{t,+\infty}\times i\mathbb{R}\subset\mathbb{C}$, where the convergence set $\mathcal{A}_{t,+\infty} \subset \mathbb{R}$ is given by
\begin{equation}
\mathcal{A}_{t,+\infty}=\left\{ u\subset\mathbb{R}:A({u})\ge0\text{ and }f_j({u})\ ({j=1,\dots,4}) \text{ satisfy \eqref{eq4.9}--\eqref{eq4.12}} \right\}\notag
\end{equation}
with
\begin{align}
f_1({u})&=\frac{\kappa au\rho}{\sigma}-a^2\left({\frac{u}{2}+l}\right)+\frac{1}{2}u^2({1-\rho^2})a^2-\frac{\kappa^2}{2\sigma^2}\le0,\label{eq4.9}\\
f_2({u})&=\frac{\kappa\theta}{\sigma^2}+\frac{1}{2}+\frac{m_u}{2}+\frac{ub\rho}{\sigma}>0,\label{eq4.10}\\
f_3({u})&=\left({\kappa\theta-\frac{\sigma^2}{2}}\right)^2-2\sigma^2\left[\frac{bu\rho({\sigma^2-2\kappa\theta})}{2\sigma}-b^2\left({\frac{u}{2}+l}\right)+\frac{1}{2}u^2({1-\rho^2})b^2\right]\ge0,\label{eq4.11}\\
f_4({u})&=\sqrt{A_u}+\kappa-ua\rho\sigma\ge0.\label{eq4.12}
\end{align}

Moreover, let
$$\mathcal{A}_{t,T}=\left\{ u\subset\mathbb{R}:A({u})\ge0\text{ and }f_j({u})\ ({j=1,\dots,3}) \text{ satisfy \eqref{eq4.9}--\eqref{eq4.11}},\ f_4({u})<0 \right\},$$
apparently $\mathcal{A}_{t,T}\supset\mathcal{A}_{t,+\infty}$. For $u \in \mathcal{D}_{t,T}=\mathcal{A}_{t,T}\times i\mathbb{R}$, the transform is well defined till the maximal (explosion) time $t+t^*$ where $t^*$ is given by
\begin{equation}
t^*=\frac{1}{\sqrt{A_u}}\log\left[1-\frac{2\sqrt{A_u}}{\kappa-\sigma ua\rho+\sqrt{A_u}}\right]
\end{equation}
\end{theorem}

\begin{lemma}\label{lem4.2}
Let $u \in \mathbb{C}$ and $l \in \mathbb{R}^+$. The distribution of $\xi$ is assumed to be normal with mean $\mu$ and variance $\eta^2$. $N$ is a Possion process at constant rate $\lambda$ and adapted to a filtration $(\mathscr{F}_t)_{\in[0,T]}$. Then
\begin{equation}
\mathbb{E}\left[ \exp\left({u\sum_{j=N_t+1}^{N_T}\xi_j-l\sum_{j=N_t+1}^{N_T}\xi_j^2}\right) \right]=\exp\left({\lambda ({T-t})({c-1})}\right),
\end{equation}
where
$$c=\frac{\exp\left({\frac{2u\mu-2l\mu^2+u^2\eta^2}{2+4l\eta^2}}\right)}{\sqrt{1+2l\eta^2}}.$$
\end{lemma}

\begin{proof}
The result follows immediately from
$$\mathbb{E}\left[\exp({u\xi_j-l\xi_j^2})\right]=\frac{\exp\left({\frac{2u\mu-2l\mu^2+u^2\eta^2}{2+4l\eta^2}}\right)}{\sqrt{1+2l\eta^2}},$$
and for $d\ge0$
$$\mathbb{E}\left[d^{N_T-N_t}\right]=\exp({\lambda ({T-t})({d-1})}).$$
\end{proof}

Despite the simplicity of the financial framework, the 4/2 plus jumps model leads to highly non trivial issues: what is the exact time to allow for the relevent transforms required by the Fourier-Laplace pricing approach? Fortunately, the following Lemma first established by Grasselli \cite{G2014} solves the problem of computing an expectation that goes beyond the known results on squared Bessel processes. Furthermore, it is easy to get the conditional generalized characteristic function of CIR process by using this Lemma. 

\begin{lemma}\label{lem4.3}
Let $X^x=\left\{ X_t^x,t \ge 0  \right\}$ denote the solution of the \eqref{eq2.2} (CIR) SDE and $X_0=x >0$ with $\kappa,\theta,\sigma >0$ and $2\kappa\theta \ge \sigma^2$ (Feller condition). Consider $\epsilon, \nu, \alpha, \gamma \in \mathbb{R}$ such that
\begin{align}
\epsilon&>-\frac{\kappa^2}{2\sigma^2},\\
\nu&\ge -\frac{\left({\kappa\theta-\frac{\sigma^2}{2}}\right)^2}{2\sigma^2},\\
\alpha&<\frac{\kappa\theta+\frac{\sigma^2}{2}+\sqrt{\left({\kappa\theta-\frac{\sigma^2}{2}}\right)^2+2\sigma^2\nu}}{\sigma^2},\\
\gamma&\ge-\frac{\sqrt{\kappa^2+2\epsilon\sigma^2}+\kappa}{\sigma^2}.
\end{align}
The following transform for the CIR process is well defined for all $t\ge0$ and is given by 
\begin{align}
\phi(t,x;\alpha,\gamma,\epsilon,\nu)&=\mathbb{E}\left[ ({X_t^x})^{-\alpha}\exp\left({-\gamma X_t^x-\epsilon \int_0^t X_t^x ds-\nu\int_0^t\frac{ds}{X_t^x}}\right)\right]\notag\\
&=\left({\frac{\beta({t,x})}{2}}\right)^{m+1}x^{-\frac{\kappa\theta}{\sigma^2}}({\gamma+K({t})})^{-\left({\frac{1}{2}+\frac{m}{2}-\alpha+\frac{\kappa\theta}{\sigma^2}}\right)}\notag\\
&\quad\times e^{\frac{1}{\sigma^2}\left({\kappa^2\theta t-\sqrt{A}x\coth\left({\frac{\sqrt{A}t}{2}}\right)+\kappa x}\right)}\frac{\Gamma\left({\frac{1}{2}+\frac{m}{2}-\alpha+\frac{\kappa\theta}{\sigma^2}}\right)}{\Gamma({m+1})}\notag\\
&\quad\times _{1}F_{1}\left({\frac{1}{2}+\frac{m}{2}-\alpha+\frac{\kappa\theta}{\sigma^2},m+1,\frac{\beta({t,x})^2}{4({\gamma+K({t})})}}\right),
\end{align} 
with
\begin{align}
m&=\frac{2}{\sigma^2}\sqrt{\left({\kappa\theta-\frac{\sigma^2}{2}}\right)^2+2\sigma^2\nu},\\
A&=\kappa^2+2\sigma^2\epsilon,\\
\beta({t,x})&=\frac{\sqrt{Ax}}{\frac{\sigma^2}{2}\sinh\left({\frac{\sqrt{A}t}{2}}\right)},\\
K({t})&=\frac{1}{\sigma^2}\left({\sqrt{A}\coth\left({\frac{\sqrt{A}t}{2}}\right)+\kappa}\right).
\end{align}

If
\begin{equation}
\gamma<-\frac{\sqrt{\kappa^2+2\epsilon\sigma^2}+\kappa}{\sigma^2},
\end{equation}
then the transform is well defined for all $t<t^*$, with
\begin{equation}
t^*=\frac{1}{\sqrt{A}}\log\left({1-\frac{2\sqrt{A}}{\kappa+\sigma^2\gamma+\sqrt{A}}}\right).
\end{equation}
\end{lemma}

\begin{remark}\label{rem4.1}
\textit{Special case}:
when $\epsilon=\nu=\gamma=0$, we have the (non-integral) moments of the process for $\alpha < \frac{2\kappa\theta}{\sigma^2}$:
\begin{align*}
\mathbb{E}\left[{X_t^{-\alpha}}\right]=&\left({\frac{\kappa}{\sigma^2}}\right)^{\alpha}\left({\sinh\left({\frac{\kappa t}{2}}\right)}\right)^{-\frac{2\kappa\theta}{\sigma^2}}\exp\left(\frac{\kappa}{\sigma^2}\left({\kappa\theta t+x-x\coth{\left({\frac{\kappa t}{2}}\right)}}\right)\right)\\
&\times\left({1+\coth\left({{\frac{\kappa t}{2}}}\right)}\right)^{\alpha-\frac{2\kappa\theta}{\sigma^2}}\frac{\Gamma\left({\frac{2\kappa\theta}{\sigma^2}-\alpha}\right)}{\Gamma\left({\frac{2\kappa\theta}{\sigma^2}}\right)}\ _{1}F_1\left({\frac{2\kappa\theta}{\sigma^2}-\alpha,\frac{2\kappa\theta}{\sigma^2},\frac{2\kappa x}{\sigma^2({e^{\kappa t}-1})}}\right).
\end{align*}
\end{remark}

\begin{proof}[Proof of Lemma \ref{lem4.3}]
The result follows immediately from Theorem 1 in Grasselli \cite{G2014}, whose proof mainly relies on Lie's classical symmetry method as in Bluman and Kumei \cite{BK2013} and Olver \cite{O2000}.We first note that by standard arguments the expectation is related to the solution of the following symmetrical PDE:
\begin{equation}
u_t=\frac{1}{2}\sigma^2xu_{xx}+f(x)u_x-\left({\frac{\nu}{x}+\epsilon x}\right)u,\quad \epsilon>0,\ \nu>0,
\end{equation}
where $f(x)=\kappa\theta-\kappa x$. The key result in order to find the Lie groups admitted by the PDE states that one should find the invariant surface for the second prolongation of group acting on the $({x;t;u})$-space where the solutions of the PDE lie. Once such equations are solved, one can find the corresponding Lie group admitted by the PDE and thus find a fundamental solution of the PDE by inverting a Laplace transform. Finally, Craddock and Lennox \cite{CL2009} showed the condition under which the fundamental solution is also a transition probability density for the underlying stochastic process. For more details, see Grasselli \cite{G2014}.
\end{proof}

\begin{proof}[Proof of Theorem \ref{thm4.1}]
From \eqref{eq2.2} we obtain
\begin{equation}
V_T-V_t=\kappa \theta \tau-\kappa\int_t^T V_s ds+\sigma\int_t^T\sqrt{V_s}dW_s,
\end{equation}
and
\begin{equation}
\log \left({\frac{V_T}{V_t}}\right)=\sigma\int_t^T\frac{1}{\sqrt{V_s}}dW_s-\kappa\tau+\left({\kappa\theta-\frac{\sigma^2}{2}}\right)\int_t^T\frac{1}{V_s}ds.
\end{equation}
Now,
\begin{align*}
Y_{t,T}&=u\log \left({\frac{\tilde{S}_T}{\tilde{S}_t}}\right)-l\int_t^T\left({a\sqrt{V_s}+\frac{b}{\sqrt{V_s}}}\right)^2ds\\
&=u({r-\lambda\tilde{\mu}})\tau-\left({\frac{u}{2}+l}\right)\int_t^T\left({a\sqrt{V_s}+\frac{b}{\sqrt{V_s}}}\right)^2ds\\
&\quad +u\rho\int_t^T\left({a\sqrt{V_s}+\frac{b}{\sqrt{V_s}}}\right)dW_s+u\sqrt{1-\rho^2}\int_t^T\left({a\sqrt{V_s}+\frac{b}{\sqrt{V_s}}}\right)dW_s^{\perp}\\
&=\left[u\left({r-\lambda\tilde{\mu}-ab-\frac{\kappa\theta a\rho}{\sigma}+\frac{\rho b \kappa}{\sigma}}\right)\tau-2lab\tau-\frac{au\rho}{\sigma}V_t-\frac{u\rho b}{\sigma}\log V_t \right]\\
&\quad +\left[ \frac{\kappa au\rho}{\sigma}-a^2\left({\frac{u}{2}+l}\right)\right]\int_t^T V_s ds\\
&\quad +\left[\frac{bu\rho\left({\sigma^2-2\kappa \theta}\right)}{2\sigma}-b^2\left({\frac{u}{2}+l}\right)\right]\int_t^T \frac{1}{V_s} ds\\
&\quad +\frac{au\rho}{\sigma}V_T+\frac{bu\rho}{\sigma}\log V_T\\
&\quad +u\sqrt{1-\rho^2}\int_t^T\left({a\sqrt{V_s}+\frac{b}{\sqrt{V_s}}}\right)dW_s^{\perp}
\end{align*}

Let $u\in \mathbb{C}$, $l \in \mathbb{R}^+$ and compute the joint Fourier-Laplace transform of $X_T$ and $(RV_T-RV_t)$. We have
\begin{align}
&\mathbb{E}\left[ \exp({uX_T-l(RV_T-RV_t)}) \bigg| V_t,X_t\right]\notag\\
&\quad=\mathbb{E}\left[ \exp({uX_t})\exp({u({X_T-X_t})-l(RV_T-RV_t)}) \bigg| V_t,X_t\right]\notag\\
&\quad=e^{uX_t}\mathbb{E}\left[ \exp\left({u\log \left({\frac{S_T}{S_t}}\right)-l\int_t^T\left({a\sqrt{V_s}+\frac{b}{\sqrt{V_s}}}\right)^2ds-l\sum_{j=N_t+1}^{N_T}\xi_j^2}\right) \Bigg| V_t,X_t\right]\notag\\
&\quad=e^{uX_t}\mathbb{E}\left[ \exp\left({u\log \left({\frac{\tilde{S}_T}{\tilde{S}_t}}\right)-l\int_t^T\left({a\sqrt{V_s}+\frac{b}{\sqrt{V_s}}}\right)^2ds}\right) \bigg| V_t,X_t\right]\notag\\
&\quad\quad\times\mathbb{E}\left[ \exp\left({u\sum_{j=N_t+1}^{N_T}\xi_j-l\sum_{j=N_t+1}^{N_T}\xi_j^2}\right) \right]\notag\\
&\quad= e^{uX_t}\mathbb{E}\left[ \exp({Y_{t,T}})\Bigg| V_t,X_t\right]\mathbb{E}\left[ \exp\left({u\sum_{j=N_t+1}^{N_T}\xi_j-l\sum_{j=N_t+1}^{N_T}\xi_j^2}\right) \right]\notag\\
&\quad= e^{uX_t} \exp\left({-\frac{au\rho}{\sigma}V_t-\frac{bu\rho}{\sigma}\log V_t}\right)\notag\\
&\quad\quad\times\exp\left(u\left({r-\lambda\tilde{\mu}-ab-\frac{\kappa\theta a\rho}{\sigma}+\frac{\rho b \kappa}{\sigma}}\right)\tau-2lab\tau+ab\tau({1-\rho^2})u^2 \right)\notag\\
&\quad\quad\times\mathbb{E}\left[ \exp\left({u\sum_{j=N_t+1}^{N_T}\xi_j-l\sum_{j=N_t+1}^{N_T}\xi_j^2}\right) \right]\notag\\
&\quad\quad\times\mathbb{E}\left[V_T^{-\alpha}\exp\left(-\gamma V_T-\epsilon\int_t^T V_sds-\nu\int_t^T\frac{1}{V_s}ds\right) \right]
\end{align}
where
\begin{align}
&\alpha=-\frac{ub\rho}{\sigma},\\
&\gamma=-\frac{ua\rho}{\sigma},\\
&\epsilon=-\frac{\kappa au\rho}{\sigma}+a^2\left({\frac{u}{2}+l}\right)-\frac{1}{2}u^2({1-\rho^2})a^2,\\
&\nu=-\frac{bu\rho({\sigma^2-2\kappa \theta})}{2\sigma}+b^2\left({\frac{u}{2}+l}\right)-\frac{1}{2}u^2({1-\rho^2})b^2.
\end{align}

Here we use the fact that
$$u\sqrt{1-\rho^2}\int_t^T\left({a\sqrt{V_s}+\frac{b}{\sqrt{V_s}}}\right)dW_s^{\perp}$$
conditional on the path $(V_s,\ t \le s \le T)$ is centered normal. We then prove Theorem \ref{thm4.1} after some manipulations by applying Lemmas \ref{lem4.2} and \ref{lem4.3}.
\end{proof}

According to Theorem \ref{thm4.1}, we derive a closed-form solution for this joint Fourier-Laplace transform so that equity and realized-variance derivatives can be priced. In particular, for equity derivatives, using the methodology presented in Carr and Madan \cite{CM1999} and Lewis \cite{L2000}, it is possible to price European options through a Fourier inversion to the characteristic function of $X_T=\log S_T$ which is exponentially affine in $X_T$. The Fourier-cosine expansion methodology in Fang and Oosterlee \cite{FO2008} is more suitable to deal with exponentially affine type. 
On the other hand, for realized-variance derivatives, the payoff of call options on realized-variance with strike $K$ and maturity $T$ is defined as $$\left(\frac{1}{T}RV_T-K\right)^+$$. Carr and Madan \cite{CM1999}, and Carr, Geman, Madan, and et al. \cite{CGMY2005} showed that the Fourier transform of the call option on realized-variance can be easily derived in closed-form expression as a function of the log strike $k=\log{K}$. By using a sequence of strikes simultaneously, FFT method as fast numerical Laplace inversion algorithms presented by Cooley and Tukey \cite{CT1965} or the more robust control variate method developed in Drimus \cite{D2012} can be put to use. 

\section{Quasi-Closed-Form Solution for Future and Call Option Price Under the 4/2 Model}

Under a continuity assumption that the price process is replicated with the log contract (see Whaley \cite{W1993}, Demeterfi, Derman, Kamal, and et al. \cite{DDK1999}, and  Carr and Wu \cite{CW2006}), the expected quadratic variation of the log returns over the next 30 days can interpret the squared VIX index. In this section, we shall extend a more general VIX future and European call options pricing formula (prices of put options followed by put-call parity). Our result not only extends the result of Zhang and Zhu \cite{ZZ2006} by turning off the jumps and setting the jump intensity $\lambda=0$ and parameter $b=0$, but also extends VIX derivatives pricing formula of the Baldeaux and Badran \cite{BB2014} by setting $a=0$. Recall the VIX formula in CBOE\cite{CBOE2003}, the squared VIX index in Eq. (1.1) is an approximation:
\begin{equation}\label{eq5.1}
\mathrm{VIX}_t^2\approx-\frac{2}{\tau}\mathbb{E}\left[{\log \left({\frac{S_{t+\tau}}{S_te^{r\tau}}}\right)\bigg|\mathcal{F}_t}\right]\times100^2,
\end{equation}
with $\tau=\frac{30}{365}$ and $S_te^{r\tau}$ being forward price of SPX observed at time $t$ with $t+\tau$ as maturity. The following derivation of VIX options pricing formula, which is an extension of Proposition 3.4 in Baldeaux and Badran \cite{BB2014}, also extends Proposition 1 in Zhang and Zhu \cite{ZZ2006}.   

\begin{lemma}\label{lem5.1}
Let $S$, $V$, and $\mathrm{VIX}^2$ be defined by Eqs. \eqref{eq2.1}, \eqref{eq2.2} and \eqref{eq5.1}, then
$$\mathrm{VIX}_t^2=100^2\times \left({H_1+\int_0^{\tau}H_2 du}\right),$$
where
\begin{align}
H_1=&2\lambda({\tilde{\mu}-\mu})+2ab+\frac{a^2}{\tau}\left({\theta\tau+\frac{\theta-x}{\kappa}\left({e^{-\kappa\tau}-1}\right)}\right),\\
H_2=&\frac{b^2\kappa}{\tau\sigma^2}\frac{\Gamma\left({\frac{2\kappa\theta}{\sigma^2}-1}\right)}{\Gamma\left({\frac{2\kappa\theta}{\sigma^2}}\right)}\left({\sinh{\left({\frac{ku}{2}}\right)}}\right)^{-\frac{2\kappa\theta}{\sigma^2}}\exp\left({\frac{\kappa}{\sigma^2}\left({\kappa\theta u+x-x\coth\left({\frac{\kappa u}{2}}\right)}\right)}\right)\notag\\
&\times\left({1+\coth\left({\frac{\kappa u}{2}}\right)}\right)^{1-\frac{2\kappa\theta}{\sigma^2}}\ _{1}F_{1}\left({\frac{2\kappa\theta}{\sigma^2}-1,\frac{2\kappa\theta}{\sigma^2},\frac{2\kappa x}{\sigma^2({e^{\kappa u}-1})}}\right).
\end{align}
\end{lemma}

\begin{proof}
It follows that
$$\mathrm{VIX}_t^2=100^2\times\left({\frac{\mathrm{g}({V_t,\tau})}{\tau}+2\lambda({\tilde{\mu}-\mu})}\right),\ t\ge0,$$
where
\begin{align*}
\mathrm{g}({x,\tau})&=-\frac{\partial}{\partial l}\mathbb{E}\left[{\exp\left({-l\int_t^{t+\tau}\left({a\sqrt{V_s}+\frac{b}{\sqrt{V_s}}}\right)^2 ds\bigg|V_t=x}\right)}\right]\Bigg|_{l=0}.
\end{align*}
We then get the result after some manipulations by applying Remark \ref{rem4.1} with $\alpha=1$ and CIR property.
\end{proof}

Lemma \ref{lem5.1} will help us obtain the distribution of $\mathrm{VIX}_t^2$ if we know the distribution of $V_t$ for $t \ge0$. In other words, we can get the pricing formula of VIX future and option if the problem of finding the transition density function for the variance process is solved. In the risk-neutral measure, Cox, Ingersoll, and Ross \cite{CIR1985} proved that the transitional probability density function (TPDF) of the instantaneous variance in \eqref{eq2.2} can be presented as
\begin{equation}
f^{Q}_{V_T | V_t}({y})=\frac{e^{\kappa\left(T-t\right)}}{2c}\left({\frac{ye^{\kappa(T-t)}}{V_t}}\right)^{q/2}\exp\left({-\frac{V_t+ye^{\kappa(T-t)}}{2c}}\right)I_q\left({\frac{\sqrt{V_t\cdot ye^{\kappa(T-t)}}}{c}}\right)\mathds{1}_{\left\{y \ge 0\right\}},
\end{equation}
where $c=\frac{\sigma^2({e^{\kappa(T-t)}-1})}{4\kappa}$, $q=\frac{2\kappa\theta}{\sigma^2}-1$, and $I_q({\cdot})$ is the modified Bessel function of the first kind of order $q$. The noncentral Chi-square, $\chi^2\left({2q+2,\frac{V_t}{c};\frac{ye^{\kappa(T-t)}}{c}}\right)$, with ${2q+2}$ degrees of freedom and parameter of noncentrality $\frac{V_t}{c}$, is the distribution function. Accorrding to Zhang and Zhu \cite{ZZ2006}, we extend the future price at time $t$ and being at maturity data $T$ directly as
\begin{align}\label{5.3}
\text{F}[{\mathrm{VIX}_t,t,T}]=&e^{-r({T-t})}\mathbb{E}\left[{\mathrm{VIX}_T\big| \mathcal{F}_t}\right]\notag\\   %\notag 让这一行没有标号
=&e^{-r({T-t})}\int_0^{\infty}\sqrt{100^2\times\left({\frac{\mathrm{g}({y,\tau})}{\tau}+2\lambda({\tilde{\mu}-\mu})}\right)}\  f^{Q}_{V_T | V_t}({y})dy\notag\\
=&e^{-r({T-t})}\int_0^\infty100\ \sqrt{H_1+\int_0^\tau H_2du}\  f^{Q}_{V_T | V_t}({y})dy,
\end{align}
and the price of European call option with $K$ as exercise price equals
\begin{align}\label{5.4}
\text{C}[{\mathrm{VIX}_t,K,t,T}]=&e^{-r({T-t})}\mathbb{E}\left[{\left({\mathrm{VIX}_T-K}\right)^+\big| \mathcal{F}_t}\right]\notag\\
=&e^{-r({T-t})}\int_0^{\infty}\left({\sqrt{100^2\times\left({\frac{\mathrm{g}({y,\tau})}{\tau}+2\lambda({\tilde{\mu}-\mu})}\right)}-K}\right)^+\  f^{Q}_{V_T | V_t}({y})dy\notag\\
=&e^{-r({T-t})}\int_0^\infty\left({100\ \sqrt{H_1+\int_0^\tau H_2du}-K}\right)^+\cdot f^{Q}_{V_T | V_t}({y})dy
\end{align}
%\cite{brigo2007interest}
%\cite{cox1985theory}
%\cite{craddock2009fundamental}
%\cite{craddock2009calculation}
%\cite{friedman2013partial}

\section{Data}
VIX was first introduced by CBOE in 1993 and was used to measure the market's expectation of 30-days volatility. Until 2003 the CBOE had used an updated methodology to calculate this index. The historical ``new'' VIX time series since 2003 is shown in Figure 1. All data used in our analysis can be downloaded from the Website of CBOE.

\begin{figure}[htpb]
\caption{Plot of the VIX index (12/01/2003--12/31/2014) }
\label{Figure 1}
\includegraphics[width=0.8\textwidth]{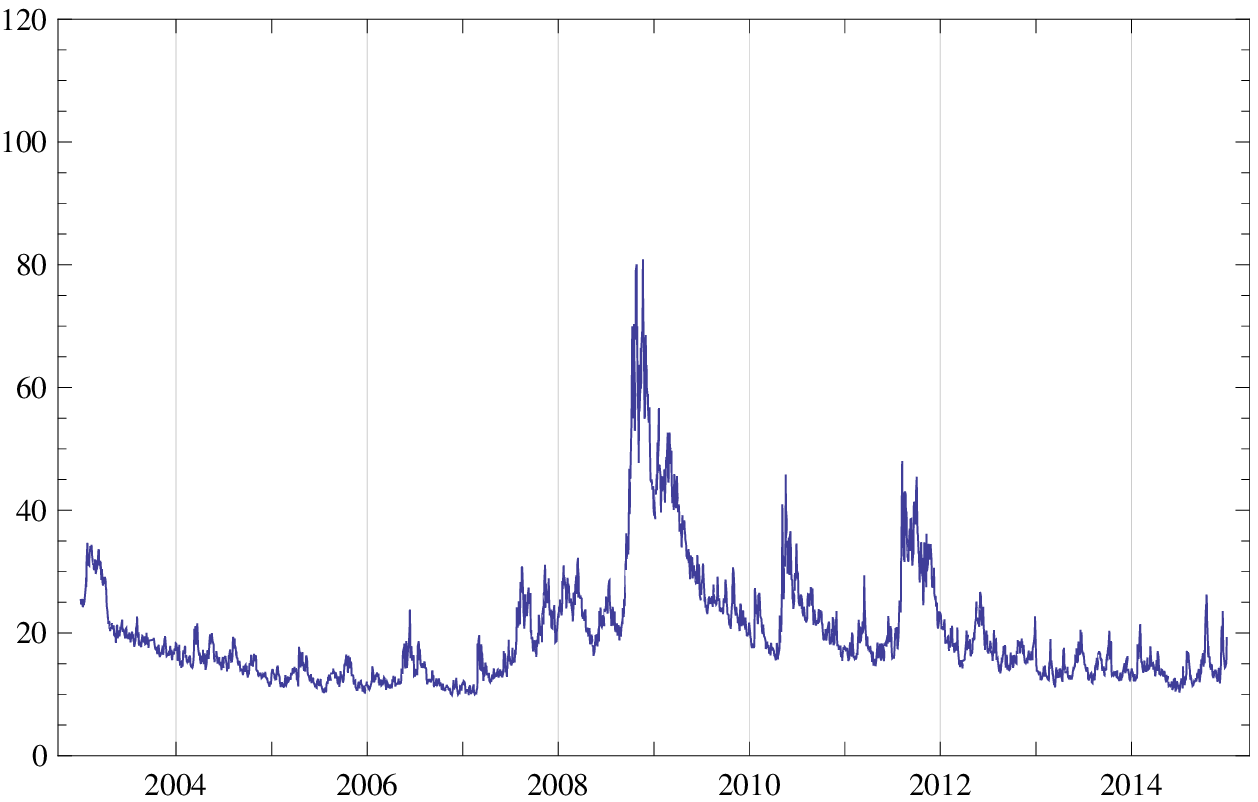}
%\begin{flushleft}
%\begin{quote} %引用
%gggg gggg gggg gggg gggg gggg gggg gggg gggg gggg gggg gggg gggg gggg gggg gggg gggg gggg gggg gggg gggg gggg gggg gggg    *图片下面可以加ggggggg这段文字*
%\end{quote}
%\end{flushleft}
\end{figure}

\begin{table}[h]  %[h]  h:放在此处 t:放在顶端 b:放在底端
\caption{VIX statistics}
\label{Table 1}
\centering
\begin{tabular}{p{3cm}l}
\hline
Symbol & Unitis\\
\noalign{\global\arrayrulewidth1pt}\hline\noalign{\global\arrayrulewidth0.4pt}\\
Mean & 19.8135\\
Variance & 87.6261\\
Minimum & 9.89\\
Maximum & 80.86\\
\hline
\end{tabular}
\end{table}

We can learn from Figure \ref{Figure 1} that VIX has a wavelike motion around a mean of approximately 20\% and tends to stay within a relatively narrow range of values suggesting a mean-reverting nature. Apparently, it can be seen that how the financial crisis starting in 2008 has been followed by a significant increase in volatility of the VIX index and served to justify its classification as the ``fear gauge". Table \ref{Table 1} shows some standard statistics for VIX data.

Note that $\mathrm{VIX}^{Quoted}$ and $\mathrm{VIX}^{Model}$ are different. In the following analysis, we refer to their relationship as 
$\left(\frac{\mathrm{VIX}^{Quoted}}{100}\right)^2=\mathrm{VIX}^{Model}.$

The option prices used in this paper are the VIX options traded in the CBOE. We employ the delayed market quotes on March 13, 2014 as the in-sample data to calibrate the risk-neutral parameters, with the underlying price 16.22, and those on March 14, 2014 are used for the out-of-sample test, with the underlying price 17.82. Note that the closing hour of the options and the VIX are the same, thus there is no nonsynchronous issue here. The whole data are available on the CBOE. To ensure sufficient liquidity and alleviate the influences of price discreetness during the valuation, we preclude the option quotes that are lower than $\$1.5$ in the sample data. $5\%$ is chosen to be the risk free interest rate.

\begin{table}
\caption{\\Description of VIX option data. The reported numbers are respectively the average option price and the numbers of observations, which are shown in the parenthesis, for the overall sample and each moneyness category on March 13, 2014 and March 14, 2014. $S$ denotes the VIX and $K$ is the exercise price of the option contract. OTM, ATM, ITM denote Out-of-the-Money, At-the-Money, In-the-Money options, respectively.}
\label{Table 1.5}
\centering
\begin{tabular*}{\textwidth}{@{\extracolsep{\fill}}lp{0.9em}lp{0.9em}lp{0.9em}lp{0.9em}l}
\hline
Date & &  Total & & \multicolumn{5}{c}{Moneyness $S/K$}\\ 
\cline{5-9}
& & & & OTM & &ATM & &ITM \\
& & & & $\le-0.1$ & & $(-0.1,0.1)$ & & $\ge0.1$ \\
\hline
March 13, 2014 & & \$3.21 & & $\$2.06$ & & $\$2.82$ & & $\$4.63$ \\
                          & &  (53)    & &      (18)   & &    (16)      & &   (19)\\
March 14, 2014 & & \$3.26 & & $\$1.86$ & & $\$2.44$ & & $\$4.24$ \\
                          & &  (57)    & &      (13)   & &    (14)      & &   (30)\\
\hline
\end{tabular*}
\end{table}

Finally the option sample contains 53 call options on March 13, 2014 and 57 call options on March 14, 2014, respectively, with available maturities: March 18, 2014; April 16, 2014; May 21, 2014; June 18, 2014; July 16, 2014; October 20, 2014. We divide the option data into 3 categories according to the moneyness $S/K$, where $S$ and $K$ denote respectively the VIX and the exercise price: Out-of-the-Money (OTM), At-the-Money (ATM), In-the-Money (ITM). As shown in Table \ref{Table 1.5}, we describe the sample by exhibiting the average prices and corresponding sample size for each moneyness category. Note that there are totally 110 call options with ITM, ATM, OTM options taking up $44.5\%$, $27.2\%$ and $28.1\%$ respectively.

\section{Model Estimation and Testing}
The VIX index and VIX options both contain information about the future dynamics of VIX index in 4/2 model. Thus, calibrating to both index and utilizing the combined informational content makes empirical sense. First the model are calibrated to the VIX index by minimizing the following objective function: 
\begin{equation}\label{eq7.1}
S_1=\sum_{n=1}^{N_1}\abs{\left(\frac{\mathrm{VIX}^{Quoted}_n}{100}\right)^2-\mathrm{VIX}^{Model}_n}
\end{equation}
using a gradient-based minimization algorithm where $N_1$ is the number of the VIX data for the 2014 one-year period. Minimizing the function in Eq. (\ref{eq7.1}) can lead to different VIX-calibrated optimal parameters. Second, let those VIX-calibrated optimal parameters be starting parameters. The model are calibrated to the VIX options by minimizing the follow objective function:
$$S_2=\sum_{n=1}^{N_2}\abs{C_n-\hat{C}_n}$$
using a gradient-based minimization algorithm where $N_2$ is the number of the sample data, $C_n$ and $\hat{C}_n$ represent the market price and the 4/2 model price respectively. Then, the optimum parameters can contain both VIX index and VIX option information. One should note that the penalty function is needed to be included in our estimation for satisfying Feller condition. The same estimation method can also be applied to the Heston model and the 3/2 model by taking $b=0$ and $a=0$ respectively. In order to test how close the model is fitted to the data point set, we compare each model's performance by VIX average relative pricing errors (VIXARPE) and option average relative pricing errors (OPTARPE): 
$$\mathrm{VIXARPE}=\frac{1}{N_1}\sum_{n=1}^{N_1}\frac{\abs{\mathrm{VIX}_n-\widehat{\mathrm{VIX}}_n}}{\mathrm{VIX}_n}$$
$$\mathrm{OPTARPE}=\frac{1}{N_2}\sum_{n=1}^{N_2}\frac{\abs{C_n-\hat{C}_n}}{C_n}$$
which is a measure to report the average error in percentage. These figures can be used to compare the models in terms of the explanatory power. Here Table \ref{Table 2} gives estimation result of the methods.

\begin{table}[h]
\caption{\\Estimates of risk-neutral parameters. By minimizing the sum of the absolute errors between the VIX index and the model-determined VIX index for 2014 one-year period, and by minimizing the sum of the absolute errors between the market price and the model-determined price for each option on March 13, 2014,  the estimated parameters for a given model are reported. VIXARPE and OPTARPE in the given row groups display the mean absolute percentage errors for VIX index and options, respectively.}
\label{Table 2}
\centering
\begin{tabular*}{\textwidth}{@{\extracolsep{\fill}}lllll}
\hline
Parameters &  4/2 Model & 3/2 Model & Heston Model & BS Model\\
\hline
$\kappa$                &  3.893244    &  2.461431     &3.848760   &\\
$\theta$                  &  0.232984     &  45.452891  &0.040210 &\\
$\sigma$                &   0.445445  &  -9.249878     &0.429494 &  0.613439\\
$a$                        &   0.9914564     &  0               &    1 \\
$b$                        &    0.180281    &   1                &   0  \\
$\lambda$             &   0.141478    &  0.143312     &0.429609\\
$\mu$                    &   -0.141627   &   -0.000067   &-0.000015\\
$\eta$                    &   0.178443   &   0.000072      &  0.0000127 \\
\\
$\mathrm{VIXARPE}$               &   14.22\%  &   12.52\%  & 25.41\%   \\
$\mathrm{OPTARPE}$               &   14.25\%  &   14.84\%  & 15.49\%  &  26.15\%\\
\hline
\end{tabular*}
\end{table}

Parameters estimated from the historical VIX data for the recent one-year period are closer to the VIX future market price (Zhang and Zhu \cite{ZZ2006}) so that those parameters are better candidates for pricing future. We will have a discuss in the following section. In Table \ref{Table 2}, we add Black Scholes model termed BS, which is regarded as a benchmark to compare with other models. It can be seen that 4/2 model has the lowest error OPTAPRE. This supports 4/2 model as a better model than Heston and 3/2 model for describing the behavior of the VIX option.

\section{Testing the VIX Future and Option Formula}
In this section, we provide a comparison of VIX derivatives market prices with our three models prices. Futures are compared by using several different kinds of maturities. In addition, we plot future price as a function of time to maturity. As for option, we assess and investigate the model performances from two angles: 1) in-sample pricing errors, 2) out-of-sample pricing errors. The following ARPE $$\mathrm{ARPE}=\frac{1}{N}\sum_{i=1}^{N}\frac{\abs{Q_i^{Market}-Q_i^{Model}}}{Q_i^{Market}}$$ error measure reports the average pricing error in percentage. Notice that $Q_i$ is used to denote quotes on futures and call options.

At first, we concentrate on four alternative models to price futures: 4/2 model, 3/2 model and Heston model. We choose four kinds of VIX future on March 13, 2014 traded in the CBOE (VIX/H14, VIX/J14, VIX/K14, VIX/M14) to test three future price formulae. Here four different futures represent four different maturities: March 14, 2014; April 14, 2014; May 14, 2014; June 14, 2014 with time to maturities being 1, 32, 62 and 93 days, respectively. The VIX level at March 13, 2014 was 16.22. For each future, we calculate corresponding future average relative pricing errors (APRE) presented in Table \ref{Table 3}. To get a sense of the capability of each model capturing future price, we plot those of the VIX future value as a function of time to maturity under Heston, 3/2 and 4/2 model (Figure \ref{Figure 2}(a-c)), using parameter values as given in Table \ref{Table 2}. The last graph in Figure \ref{Figure 2} reflects the performance of three models. Here time to maturity is  annualized in Figure \ref{Figure 2}. 

\begin{table}
\caption{\\VIX future pricing errors. For given 4/2, 3/2 and Heston model, we compute the price of each future on March 13, 2014, with a total of 9 futures, using the parameters in Table \ref{Table 2}. The group under the heading ARPE reports the average pricing error between the market price and the model price for each future. Total ARPE in the given last row group display the total ARPE with four futures.}
\label{Table 3}
\centering
\begin{tabular*}{\textwidth}{@{\extracolsep{\fill}}llll}
\hline
Future & \multicolumn{3}{c}{ARPE}\\ \cline{2-4}
& 4/2 Model & 3/2 Model & Heston Model \\
\hline
VIX/H14 & 0.05\% & 0.01\% & 0.48\% \\
VIX/J14 & 2.49\% & 3.68\% & 2.73\% \\
VIX/K14 & 1.45\% & 2.87\% & 1.44\% \\ 
VIX/M14 & 0.58\% & 1.05\% & 0.43\% \\ 
\\
Total ARPE & 1.14\% & 1.90\% & 1.27\% \\ 
\hline
\end{tabular*}
\end{table}

\begin{figure}[htpb]
\caption{\\Graphs of the VIX future price as a function of time to maturity (years). VIX futures on March 13, 2014 to November 14, 2014 are computed by three models, using the parameters estimated in Table \ref{Table 2}.}
\label{Figure 2}
\subfigure[4/2 Model]{\label{4/2 Model}\includegraphics[width=0.4\textwidth]{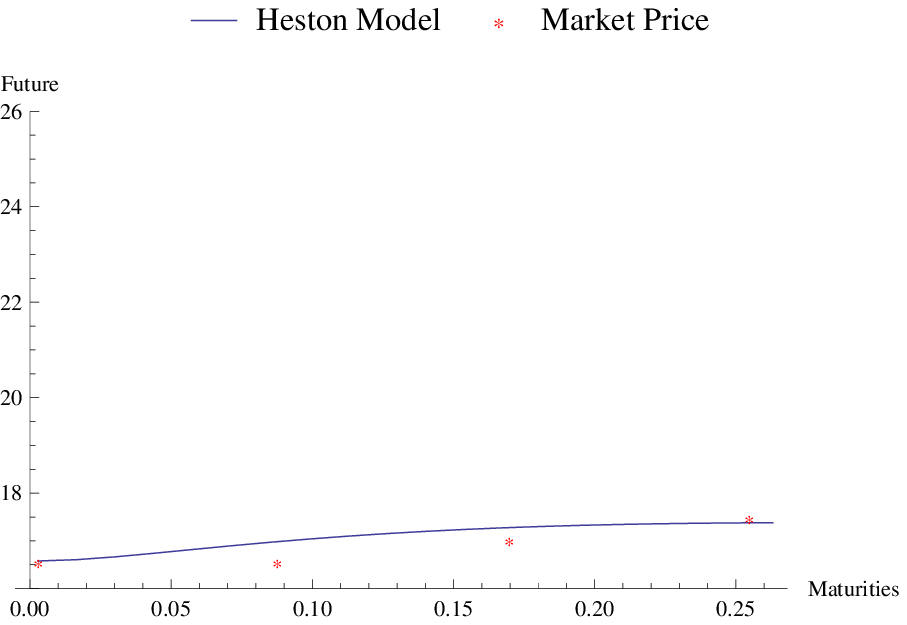}}
\subfigure[3/2 Model]{\label{3/2 Model}\includegraphics[width=0.4\textwidth]{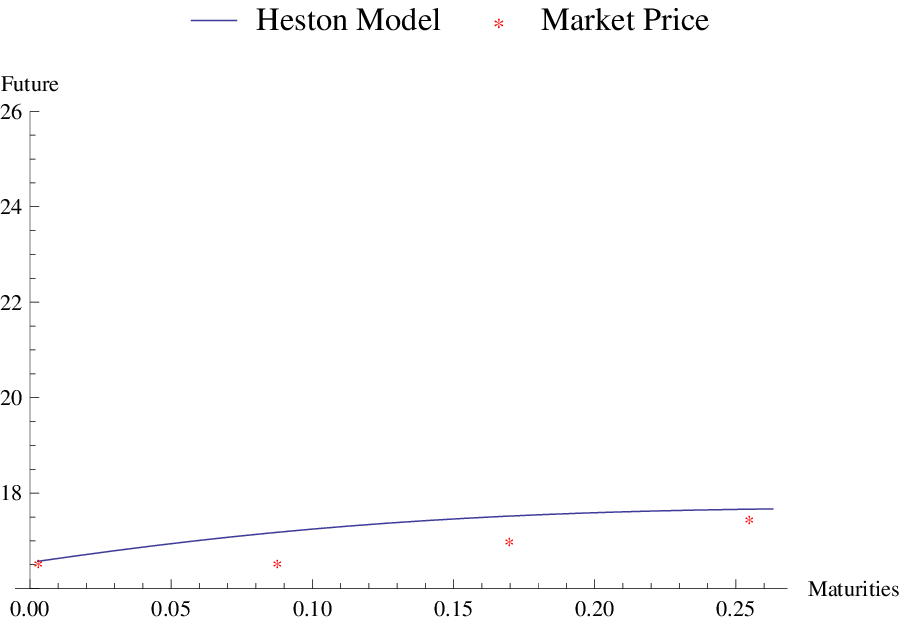}}
\subfigure[Heston Model]{\label{Heston Model}\includegraphics[width=0.4\textwidth]{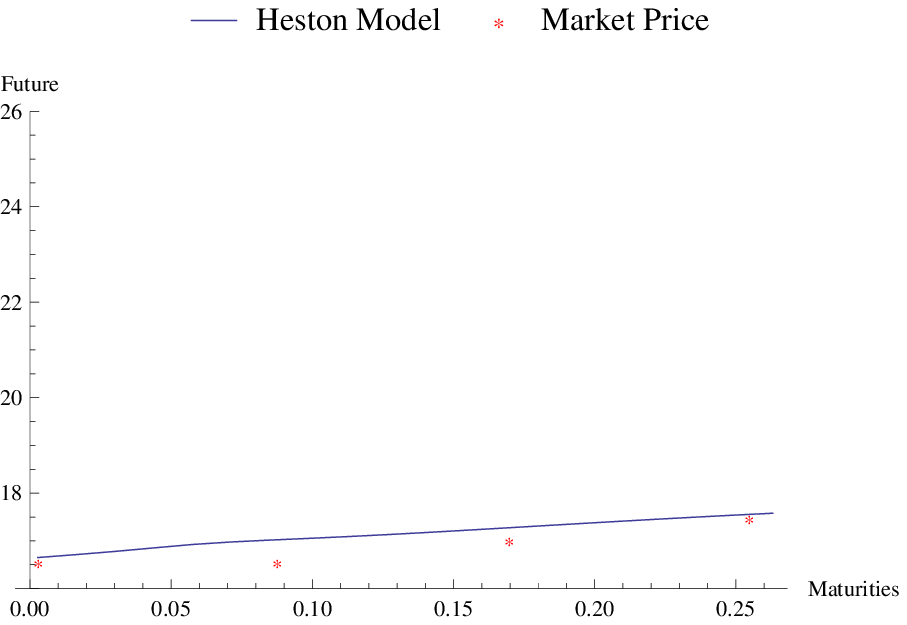}}
\subfigure[Three Models]{\label{Three Models}\includegraphics[width=0.42\textwidth]{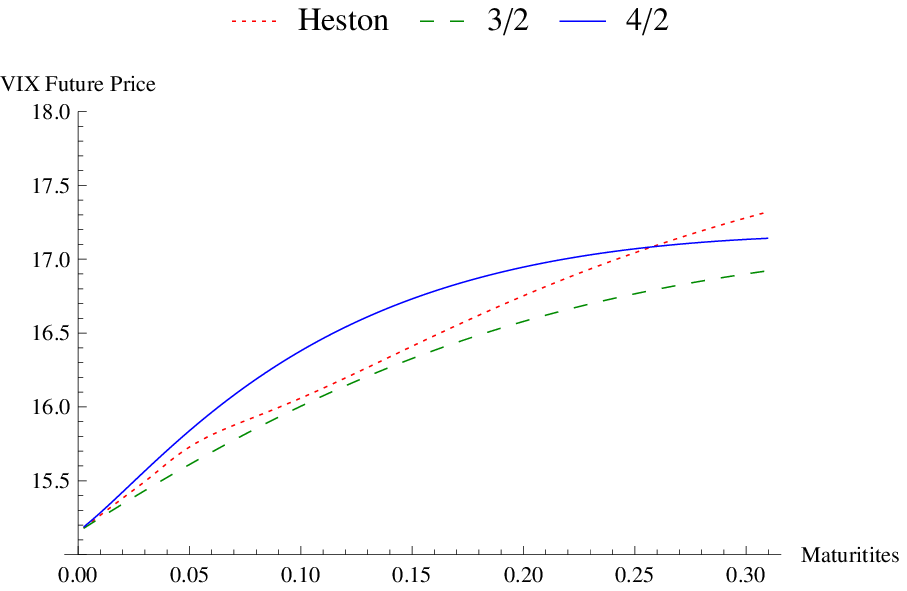}}
\end{figure}

It can be seen from Table \ref{Table 3} and Figure \ref{Figure 2} that all models have excellent performance in pricing VIX future. They have some common basic features. Firstly under each model the VIX future roughly increase with the future maturities. Secondly, we note that all model futures with maturities of 1 day, 32 days, 62 days and 93 days have very low ARPE. Finally, we can see from Figure \ref{Three Models} that 4/2 model has new features. Compared with other models, 4/2 model has a new feature of more upward-bulging curvature, with $-\frac{\partial^2F}{\partial t^2}$ taking a larger value. After 80 maturities, the curve of 4/2 model tends to be flat to protect the VIX future from unrealistically ascending.

On the other hand, we turn our concentrations on pricing options, with 4/2 model, 3/2 model and Heston model, using option data presented in Table \ref{Table 1.5}. The results are presented in two tables and two graphs. The graphs in Figure \ref{Figure 4} and Figure \ref{Figure 4.1} are VIX option price as a function of strikes and VIX values, respectively. Those graphs are quite informative to explore reactions of each model when facing extreme situation in strikes or VIX values. The statistics in Table \ref{Table 4} and Table \ref{Table 5} are respectively in-sample pricing errors and corresponding out-of-sample pricing errors. In-sample pricing errors are quite informative to explain the internal working of each model. Moreover, out-of-sample pricing errors help us understanding predictive qualities of each model.

\begin{figure}[htpb]
\caption{\\Graphs of the VIX option price as a function of strikes. VIX options on March 13, 2014 with 33 different strikes are computed by three models, using the parameters estimated in Table \ref{Table 2}. Notice that the strike in graphs is the fact divided by 100.}
\label{Figure 4}
\subfigure[Heston model]{\label{VIX1422A13}\includegraphics[width=0.32\textwidth]{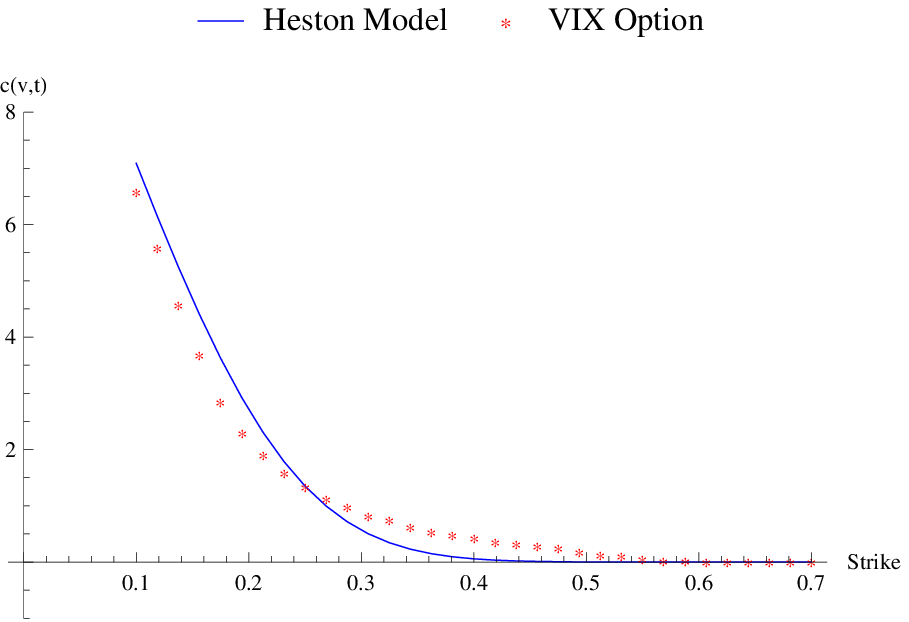}}
\subfigure[3/2 model]{\label{VIX1419B25}\includegraphics[width=0.32\textwidth]{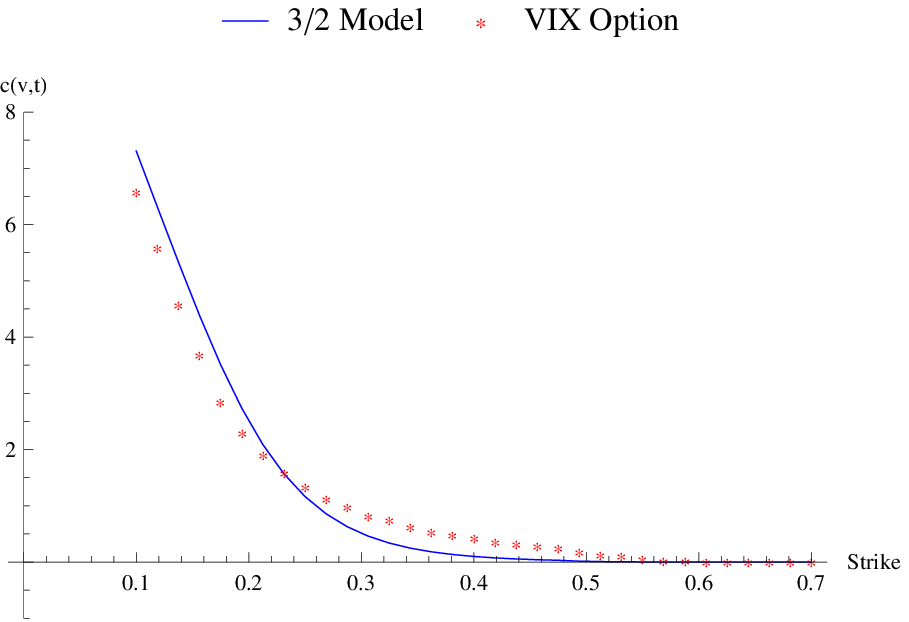}}
\subfigure[4/2 model]{\label{VIX1415D35}\includegraphics[width=0.32\textwidth]{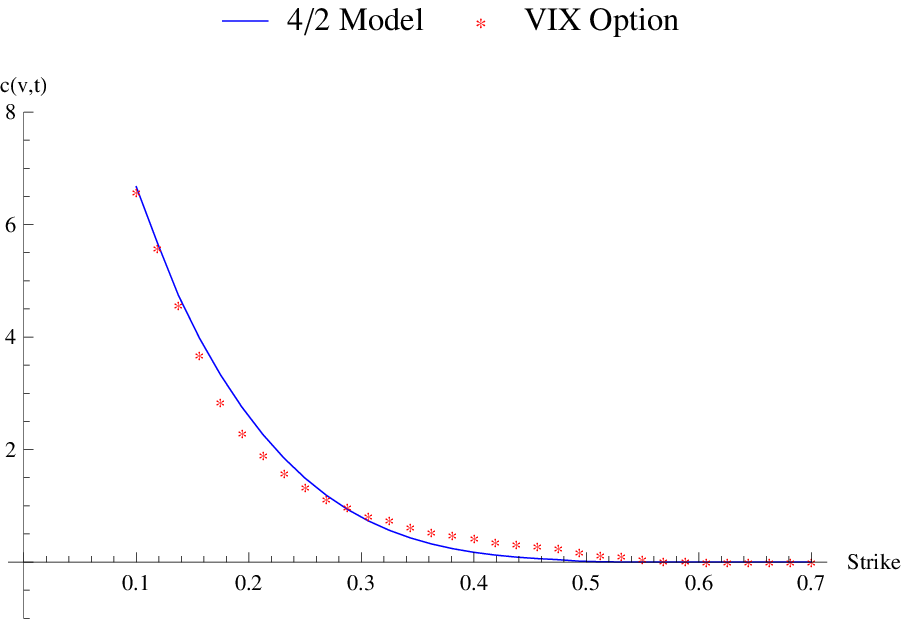}}
\end{figure}

To get a sense of the capability of each model capturing features of option with different strikes we choose 33 VIX options traded on March 13, 2014 with same maturity and different strikes from 10 to 70 and plot those of options respectively with three models as reflected in Figure \ref{Figure 4}. Regarding three comparative graphs, we find that 4/2 model can best fit the model price, especially at the lower strikes. On the other hand, however, at the higher strikes, all distributions implied by three models are more or less to understate the probability of extreme movement in the VIX. From Figure \ref{Figure 1} it is apparent that the VIX can exhibit spikes making large movements in the VIX possible. This leads to higher out-of-money call option prices. Compared with other models, 4/2 model goes down more slowly with strikes to fit market price better, especially the strikes between 28 and 40. 

Similarly, to get a sense of the capability of each model capturing features of option with different VIX values, we choose 27 different VIX values from 11.5 to 28.5 and plot those options with four fixed $t=\frac{1}{365}$, $\frac{5}{365}$, $\frac{10}{365}$ and $\frac{30}{365}$ as maturities and fixed strike $K=0.18$ as reflected in Figure \ref{Figure 4.1}. As expected, all of models have common basic feature, which the call option value increase with VIX. It can seen that for larger values of the VIX, the values of call option decrease with maturities whereas for small values of the VIX the call value increases with maturities. This is to be expected as when time approaches expiry there is less time for VIX to revert to its mean value. Heston and 3/2 models also understate the call option when dealing with small VIX values. Fortunately, 4/2 model accelerates slowly and steadily with VIX and performs more better at lower VIX.

\begin{figure}[htpb]
\caption{\\Graphs of the VIX option price as a function of VIX values. Fixed 4 maturities, VIX options with different VIX values are computed by three models, using the parameters estimated in Table \ref{Table 2}.}
\label{Figure 4.1}
\subfigure[Heston model]{\label{Heston model}\includegraphics[width=0.32\textwidth]{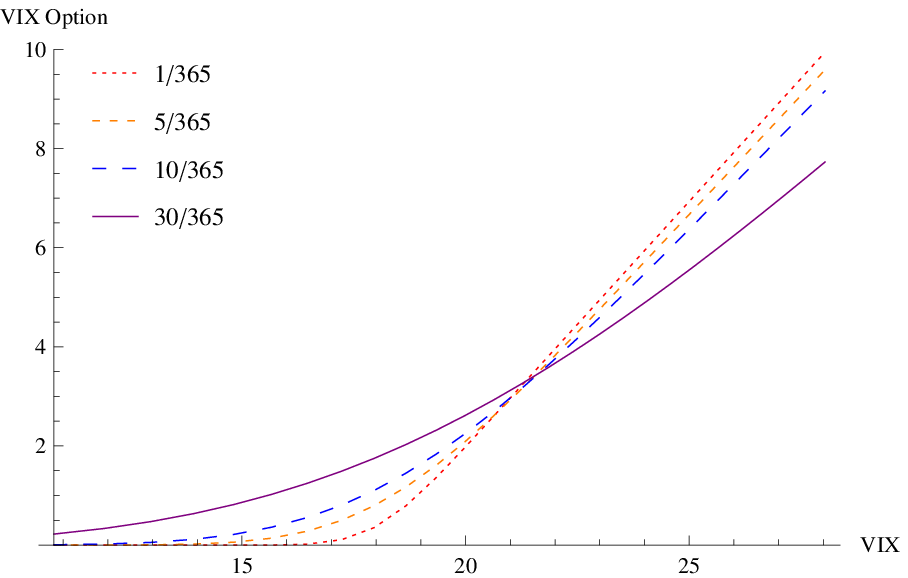}}
\subfigure[3/2 model]{\label{3/2 model}\includegraphics[width=0.32\textwidth]{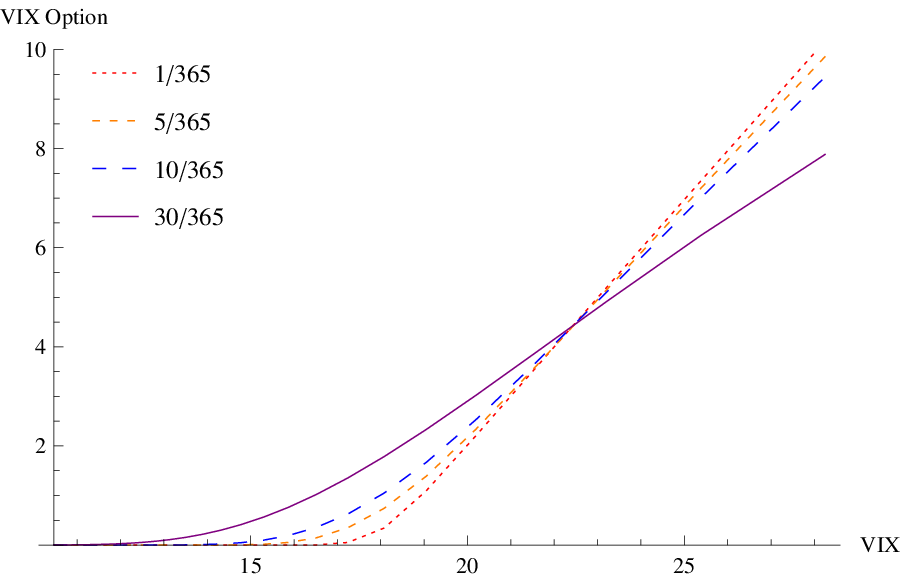}}
\subfigure[4/2 model]{\label{4/2 model}\includegraphics[width=0.32\textwidth]{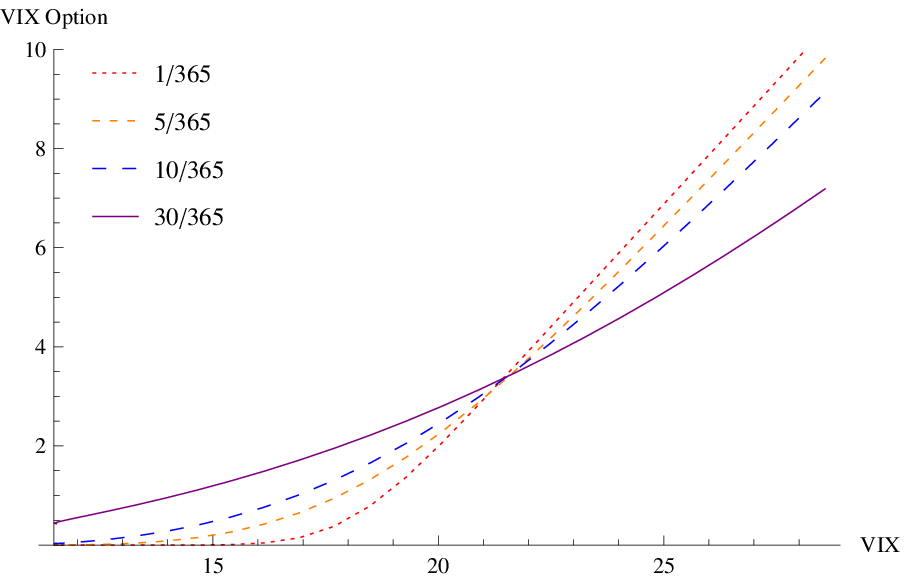}}
\end{figure}

\begin{table}
\caption{\\In-sample pricing errors. For a given model, we compute the price of each option on March 13, 2014, with a total of 53 options, using the parameters estimated in Table \ref{Table 2}. The group under the heading APRE reports the sample average pricing error between the market price and the model price for each option in a given moneyness category. Total ARPE in the given row group displays the total ARPE with 53 options.}
\label{Table 4}
\centering
\begin{tabular*}{\textwidth}{@{\extracolsep{\fill}}lp{0.9em}lp{0.9em}lp{0.9em}lp{0.9em}l}
\hline
     &  & \multicolumn{6}{c}{ARPE}\\ 
\cline{3-9}
Moneyness $M$ & & 4/2 model & & 3/2 model & & Heston model & & BS model\\
\hline
$(0.3,0.5)$ & & 5.82\% & & 6.73\% & & 6.39\% & & 7.80\%\\
(7)          &  &              &  &               & &               & &          \\
$(0.1,0.3)$ & & 9.49\% & & 9.52\%  & & 14.01\% & & 13.44\%\\
(12)          &  &              &  &               & &               & &          \\
$(-0.1,0.1)$ & & 12.39\% & & 8.12\%  & & 11.46\% & & 25.08\%\\
(16)          &  &              &  &               & &               & &          \\
$(-0.3,-0.1)$ & & 17.06\% & & 22.50\%  & & 16.99\% & & 40.85\%\\
(14)          &  &              &  &               & &               & &          \\
$(-0.5,-0.3)$ & & 40.92\% & & 44.99\%  & & 46.74\% & & 49.19\%\\
(4)          &  &              &  &               & &               & &          \\
             &  &                 & &               & &              & &           \\
Total ARPE  & & 14.25\% &  & 14.84\% &  & 15.49\% &  &   26.15\% \\             
\hline
\end{tabular*}
\end{table}

To detailedly investigate the relationship between model and market prices and whether 4/2 model performance of options is better than others or not, for each option, we calculate the corresponding models prices, ARPE and moneyness $M$, defined by $M=\ln{\frac{\mathrm{VIX}_{Market}}{Strike}}$ for the three models. Note that a positive (resp. negative) $M$ value denotes that the option is in (resp. out)-of-the-money and the larger the magnitude of $M$, the deeper it is. Results are grouped into the range of Moneyness $M$: (0.3, 0.5), (0.1, 0.3), (-0.1, 0.1), (-0.3, 0.1) and (-0.5, -0.3), which are listed in Table \ref{Table 4}. According to the magnitude of in-sample total ARPE, we find that 4/2 model can best fit the market prices. Regarding the reported ARPE value, all models deliver much small ARPE for deep ITM option and large ARPE for OTM option fitting, which is in line with the estimation method that assigns more weight to the ITM options and less weight to the OTM options. 4/2 model still has the fewest errors. However, 3/2 model model results in a slightly smaller error than 4/2 model when the options are at-the-money, due to the reason that 3/2 model fits the ATM options better than others, which can be seen in Table \ref{Table 4}, and the poor fitting in ITM and OTM options has increased the value of 3/2 model's ARPE. Compared to the other three models, BS gives much errors by looking at values of ARPE. The generations of BS all provide large pricing improvements, for the most part, in OTM and ATM options, with ARPE 49.19\% and 25.08\%, respectively. On the whole, we would like to draw the conclusion that 4/2 model shows the best in-sample performance, being capable of fitting market prices, and 3/2 model on the other hand, gives competitive performance in consideration of its fewer parameters requirements and good qualifying performance in ATM options.

\begin{table}
\caption{\\Out-of-sample pricing errors. For a given model, we compute the price of each option on March 14, 2014, with a total of 57 options, using the parameters estimated in Table \ref{Table 2}. The group under the heading APRE reports the sample average pricing error between the market price and the model price for each option in a given moneyness category. Total ARPE in the given row group displays the total ARPE with 57 options.}
\label{Table 5}
\centering
\begin{tabular*}{\textwidth}{@{\extracolsep{\fill}}lp{0.9em}lp{0.9em}lp{0.9em}lp{0.9em}l}
\hline
  &  & \multicolumn{6}{c}{ARPE}\\ 
\cline{3-9}
Moneyness $M$ & & 4/2 model & & 3/2 model & & Heston model & & BS model\\
\hline
$(0.3,0.5)$ & & 5.82\% & & 9.46\% & & 10.77\% & & 15.03\%\\
(12)          &  &              &  &               & &               & &          \\
$(0.1,0.3)$ & & 12.20\% & & 12.78\%  & & 14.67\% & & 23.44\%\\
(18)          &  &              &  &               & &               & &          \\
$(-0.1,0.1)$ & & 14.23\% & & 10.13\%  & & 8.76\% & & 35.85\%\\
(14)          &  &              &  &               & &               & &          \\
$(-0.3,-0.1)$ & & 26.20\% & & 29.41\%  & & 32.03\% & & 47.45\%\\
(13)          &  &              &  &               & &               & &          \\
             &  &                 & &               & &              & &           \\
Total ARPE  & & 14.55\% &  & 15.22\% &  & 16.36\% &  &   30.19\% \\             
\hline
\end{tabular*}
\end{table}

Now that the in-sample fit is increasingly better from BS, Heston, 3/2 and 4/2, one may argue that the outcome can be biased due to the larger number of parameters and the over-fitting to the data. Moreover, a model that performs well in fitting option prices may have poor predictive qualities. Given these concerns, we design the out-of-sample test by using parameters estimated in Table \ref{Table 2} as inputs to compute the model-based option prices on March 14, 2014 and report the corresponding ARPE in Table \ref{Table 5}. According to the results, almost all models deliver larger errors for out-of-sample option fitting than in-sample. Regarding the reported total ARPE, we find that 4/2 model generates the lowest ARPE with a little increase of ARPE compared with in-sample and can best fit the out-of-sample market prices. As for OTM options, 3/2 and Heston still generate many pricing errors and 4/2 model performs slightly better. The overall ATM options ARPE for 4/2, 3/2, Heston and BS are respectively 14.23\%, 10.13\%, 8.76\% and 35.85\%. It is accidentally surprising that for Heston model, ARPE for ATM options is 8.76\%, which is lower than 11.46\% of in-sample. 3/2 model keeps relatively lower errors in ATM options. Fortunately, 4/2 model still remains the excellent performance in matching ITM options. All the models except 4/2 model generate larger percentage errors in the out-of-sample test, which shows that 4/2 model is quite competent in out-of-sample pricing.

According to the foregoing results, we can draw conclusion that traders should choose the 4/2 model in most cases. These observations make 4/2 model a suitable candidate for modeling VIX derivatives. 

\section{Conclusion}
We have extended both the popular Heston and 3/2 models in unified framework and plus jumps by keeping their analytical tractability. A new model considers as instantaneous variance the superposition of the 1/2 and the 3/2 terms, which explains the name 4/2 model. In order to make sure that the discounted stock price is a martingale under our risk neutral probability measure we derive the conditions to guarantee this. A closed-form solution for the joint transform has been derived so that equity and realized-variance derivatives can be priced. Since derivatives on the VIX hit the market in 2004, the large quantities of literatures for a solution to the VIX derivatives have steadily grown. However, it is a vacuum for describing the dynamics of S\&P500 with 4/2 model and pricing VIX derivatives. Finally, we point out that the 4/2 model provides a nice example of application of the powerful theory of Lie symmetries so that we derive a general analytic solution for the pricing of equity, VIX future and option.

Inspired by the theoretical and practical analysis, we have estimated parameters in term of their ability to explain the dynamics of VIX and VIX option. Models of the Heston, 3/2, 4/2 and BS are compared on one-day VIX future and option price data. We would like to draw a conclusion that 4/2 model is found to have the best overall performance in pricing VIX future and option. The 3/2 model ranks next, followed by the Heston. The significant performance gap between the BS model and the remaining models suggests that all the generations of the BS are more efficient with greater improvements. Summarizing all findings, the introduction of 4/2 model results in a significant improvement in VIX derivatives pricing, suggesting that 4/2 model may be a more accurate and effective guide to traders. However, the 3/2 model itself still have the competitive performance in ATM options. 

Finally, there are several works remained which have not been discussed in this paper and we hope to figure them in the future. For example, adaptive algorithms and discretization schemes should be provided to exactly simulate the 4/2 model. Furthermore, whether a joint calibration of S\&P500 and S\&P500 option data would make 4/2 model accurately pricing S\&P500 derivatives is still unknown. The corresponding hedging test should be worth exploring. Those topics also deserve more future research and will be the subject of future work.

\subsection*{Acknowledgments}
This work is supported by the National Natural Science Foundation of China (No. 11171304 and No. 71371168) and Zhejiang Provincial Natural Science Foundation of China (No. Y6110023). We also want to thank Mark McClure of the \textit{Mathematics StackExchange} community for providing us with \textit{Mathematica} codes\footnote{See \url{http://math.stackexchange.com/a/1441092/58882}.} that cope with numerical integration of the form 
$$\int_{a_1}^{b_1}\sqrt{\int_{a_2}^{b_2}f(x,y)dx}\ g(y)dy$$
(e.g., Eqs. \eqref{5.3} and \eqref{5.4}).

\end{document}